\let\oldvec\vec 
\documentclass{llncs}
\let\vec\oldvec 

\usepackage{geometry}
\usepackage{times}

\usepackage{comment}

\usepackage{amsmath,amssymb}
\usepackage{color}
\usepackage{mathtools}
\usepackage{tikz}
\usepackage{url}
\usepackage{wrapfig}
\usepackage{enumerate}

\usepackage{listings}

\usepackage{booktabs}

\pgfdeclarelayer{background}
\pgfdeclarelayer{foreground}
\pgfsetlayers{background,main,foreground}

\usetikzlibrary{positioning,chains,fit,shapes,calc}

\newcommand{\N}{\mathbb{N}}

\newcommand{\R}{\mathbb{R}}

\usepackage{xspace}
\newcommand{\perturbed}{deviated\xspace}
\newcommand{\deviation}{deviation\xspace}
\newcommand{\stability}{stability\xspace}

\newcommand{\eSR}{\text{$\epsilon$-SR}}
\newcommand{\bDR}{\text{$\beta$-DR}}

\definecolor{myblue}{RGB}{80,80,160}
\definecolor{mygreen}{RGB}{80,160,80}

\title{Path Deviations Outperform Approximate Stability in Heterogeneous Congestion Games}

\author{Pieter Kleer\inst{1} \and Guido Sch\"afer\inst{1,2}}
\authorrunning{P. Kleer \and G.\ Sch\"afer} 

\institute{%
  Centrum Wiskunde \& Informatica (CWI), Networks and Optimization Group, Amsterdam, The Netherlands
  \and
  Vrije Universiteit Amsterdam, Department of Econometrics and Operations Research, Amsterdam, The Netherlands. \\ 
  \texttt{kleer@cwi.nl, schaefer@cwi.nl}
}

\usepackage{tikz}
\usetikzlibrary{arrows,automata,positioning}
\usepackage[latin1]{inputenc}

\usepackage{ifthen}
% ...
\newenvironment{rtheorem}[3][]{%
\noindent\ifthenelse{\equal{#1}{}}{\bf #2 #3.}{\bf #2 #3 (#1)}%
\begin{it}}{\end{it}}

\begin{document}
\sloppy

\maketitle

\begin{abstract}
We consider non-atomic network congestion games with heterogeneous players where the latencies of the paths are subject to some bounded deviations. This model encompasses several well-studied extensions of the classical Wardrop model which incorporate, for example, risk-aversion, altruism or travel time delays. Our main goal is to analyze the worst-case deterioration in social cost of a \emph{\perturbed Nash flow} (i.e., for the perturbed latencies) with respect to an original Nash flow. 

We show that for homogeneous players \perturbed Nash flows coincide with approximate Nash flows and derive tight bounds on their inefficiency. In contrast, we show that for heterogeneous populations this equivalence does not hold. We derive tight bounds on the inefficiency of both \perturbed and approximate Nash flows for \emph{arbitrary} player sensitivity distributions.
Intuitively, our results suggest that the negative impact of path deviations (e.g., caused by risk-averse behavior or latency perturbations) is less severe than approximate \stability (e.g., caused by limited responsiveness or bounded rationality). 

We also obtain a tight bound on the inefficiency of \perturbed Nash flows for matroid congestion games and homogeneous populations if the path deviations can be decomposed into edge deviations. 
In particular, this provides a tight bound on the Price of Risk-Aversion for matroid congestion games.
\end{abstract}

\section{Introduction}

In 1952, Wardrop~\cite{Wardrop1952} introduced a simple model, also known as the \emph{Wardrop model}, to study outcomes of selfish route choices in traffic networks which are affected by congestion. In this model, there is a continuum of non-atomic players, each controlling an infinitesimally small amount of flow, whose goal is to choose paths in a given network to minimize their own travel times. The latency (or delay) of each edge  is prescribed by a non-negative, non-decreasing latency function which depends on the total flow on that edge.
Ever since its introduction, the Wardrop model has been used extensively, both in operations research and traffic engineering studies, to investigate various aspects of selfish routing in networks. 

More recently, the classical Wardrop model has been extended in various ways to capture more complex player behaviors. Examples include the incorporation of uncertainty attitudes (e.g., risk-aversion, risk-seeking), cost alterations (e.g., latency perturbations, road pricing), other-regarding dispositions (e.g., altruism, spite) and player biases (e.g., responsiveness, bounded rationality).

Several of these extensions can be viewed as defining some modified cost for each path which combines the original latency with some `deviation' (or perturbation) along that path. Such deviations are said to be \emph{$\beta$-bounded} if the total deviation along each path is at most $\beta$ times the latency of that path. The player objective then becomes to minimize the combined cost of latency and  deviation along a path (possibly using different norms). 
An equilibrium outcome corresponds to a \emph{$\beta$-\perturbed Nash flow}, i.e., a Nash flow with respect to the combined cost. The deviations might be given explicitly (e.g., as in the altruism model of Chen et al.~\cite{Chen2014}) or be defined implicitly (e.g., as in the risk-aversion model of Nikolova and Stier-Moses~\cite{Nikolova2015}). Further, different fractions of players might perceive these deviations differently, i.e., players might be heterogeneous with respect to the deviations. 

Another extension, which is closely related to the one above, is to incorporate different degrees of `responsiveness' of the players. For example, each player might be willing to deviate to an alternative route only if her latency decreases by at least a certain fraction. In this context, an equilibrium outcome corresponds to an \emph{$\epsilon$-approximate Nash flow} for some $\epsilon \ge 0$, i.e., for each player the latency is at most $(1+\epsilon)$ times the latency of any other path. Here, $\epsilon$ is a parameter which reflects the responsiveness of the players. An analogue definition can be given for populations with heterogeneous responsiveness parameters. 

To illustrate the relation between \perturbed and approximate Nash flows, suppose we are given a $\beta$-\perturbed Nash flow $f$ for some $\beta \geq 0$, where the latency $\ell_P(f)$ of each path $P$ is perturbed by an arbitrary $\beta$-bounded deviation $\delta_P(f)$ satisfying $0 \leq \delta_P(f) \leq \beta l_P(f)$. Intuitively, the deviations inflate the latency on each path by at most a factor of $(1+\beta)$. Further, assume that the population is homogeneous. From the Nash flow conditions (see Section~\ref{sec:pre} for formal definitions), it follows trivially that $f$ is also an $\epsilon$-approximate Nash flow with $\epsilon = \beta$. 
But does the converse also hold? That is, can every $\epsilon$-approximate Nash flow be induced by a set of bounded path deviations? 
More generally, what about the relation between \perturbed and approximate Nash flows for heterogenous populations? Can we bound the inefficiency of these flows? 

In this paper, we answer these questions by investigating the relation between the two equilibrium notions. Our main goal is to quantify the inefficiency of \perturbed and approximate Nash flows, both for homogeneous and heterogeneous populations. To this aim, we study the (relative) worst-case deterioration in social cost of a $\beta$-\perturbed Nash flow with respect to an original (unaltered) Nash flow; we use the term \emph{$\beta$-\deviation ratio} to refer to this ratio. This ratio has recently been studied in the context of risk aversion \cite{Lianeas2016,Nikolova2015} and in the more general context of bounded path deviations \cite{Kleer2016}. 
Similarly, for approximate Nash flows we are interested in bounding the \emph{$\epsilon$-\stability ratio}, i.e., the worst-case deterioration in social cost of an $\epsilon$-approximate Nash flow with respect to an original Nash flow. 

Note that these notions differ from the classical \emph{price of anarchy} notion~\cite{Koutsoupias1999}, which refers to the worst-case deterioration in social cost of a $\beta$-\perturbed (respectively, $\varepsilon$-approximate) Nash flow with respect to an \emph{optimal} flow. While the price of anarchy typically depends on the class of latency functions (see, e.g., \cite{Chen2014,Christodoulou2011,Kleer2016,Nikolova2015} for results in this context), the \deviation ratio is independent of the latency functions but depends on the topology of the network (see \cite{Kleer2016,Nikolova2015}). 


\paragraph{Our contributions.}

The main contributions of this paper are as follows: 

\begin{enumerate}\itemsep7pt

\item We show that for homogeneous populations the set of $\beta$-\perturbed Nash flows coincides with the set of $\epsilon$-approximate Nash flows for $\beta = \epsilon$. Further, we derive an upper bound on the $\epsilon$-\stability ratio (and thus also on the $\epsilon$-\deviation ratio) which is at most $(1+\epsilon)/(1-\epsilon n)$, where $n$ is the number of nodes, for single-commodity networks. We also prove that the upper bound we obtain is tight for \emph{generalized Braess graphs}. These results are presented in Section~\ref{sec:approx}. 

\item We prove that for heterogenous populations the above equivalence does not hold. We derive tight bounds for both the $\beta$-\deviation ratio and the $\epsilon$-\stability ratio for single-commodity instances on series-parallel graphs and arbitrary sensitivity distributions of the players. To the best of our knowledge, these are the first inefficiency results in the context of heterogenous populations which are tight for \emph{arbitrary} sensitivity distributions. 
Our bounds show that both ratios depend on the demands and sensitivity distribution $\gamma$ of the heterogenous players (besides the respective parameters $\beta$ and $\epsilon$). Further, it turns out that the $\beta$-\deviation ratio is always at most the $\epsilon$-\stability ratio for $\epsilon = \beta \gamma$.
These results are given in Section~\ref{sec:het}. 

\item We also derive a tight bound on the $\beta$-\deviation ratio for single-commodity matroid congestion games and homogeneous populations if the path deviations can be decomposed into edge deviations. To the best of our knowledge, this is the first result in this context which goes beyond network congestion games. In particular, this gives a tight bound on the Price of Risk-Aversion \cite{Nikolova2015} for matroid congestion games. 
This result is of independent interest and presented in Section~\ref{sec:approx}. 

\end{enumerate}

In a nutshell, our results reveal that for homogeneous populations there is no quantitative difference between the inefficiency of \perturbed and approximate Nash flows in the worst case. In contrast, for heterogenous populations the $\beta$-\deviation ratio is always at least as good as the $\epsilon$-\stability ratio with $\epsilon = \beta \gamma$. Intuitively, our results suggest that the negative impact of path deviations (e.g., caused by risk-averse behavior or latency perturbations) is less severe than approximate  \stability (e.g., caused by limited responsiveness or bounded rationality). 

\paragraph{Related work.} 

We give a brief overview of the works which are most related to our results. 
Christodoulou et al. \cite{Christodoulou2011} study the inefficiency of approximate equilibria in terms of the price of anarchy and price of stability (for homogeneous populations). Generalized Braess graphs were introduced by Roughgarden \cite{Roughgarden2006} and are used in many other lower bound constructions (see, e.g., \cite{Englert2008,Kleer2016,Roughgarden2006}). 
Chen et al. \cite{Chen2014} study an altruistic extension of the Wardrop model and, in particular, also consider heterogeneous altruistic populations. They obtain an upper bound on the ratio between an altruistic Nash flow and a social optimum for parallel graphs, which is tight for two sensitivity classes. It is mentioned that this bound is most likely not tight in general. Meir and Parkes \cite{Meir2014Arxiv} study player-specific cost functions in a smoothness framework \cite{Roughgarden2015}. Some of their inefficiency results are tight, although none of their bounds seems to be tight for arbitrary sensitivity distributions. Matroids have also received some attention in the Wardrop model. In particular, Fujishige et al. \cite{Fujishige2015} show that matroid congestion games are immune against the Braess paradox (and their analysis is tight in a certain sense). We refer the reader to \cite{Kleer2016} for additional references and relations of other models to the bounded path deviation model considered here.

\section{Preliminaries}\label{sec:pre}



Let $\mathcal{I} = (E,(l_e)_{e \in E},(\mathcal{S}_i)_{i \in [k]},(r_i)_{i \in [k]})$ be an instance of a non-atomic congestion game. Here, $E$ is the set of resources (or edges, or arcs) that are equipped with a non-negative, non-decreasing, continuous latency function $l_e: \mathbb{R}_{\ge 0} \rightarrow \mathbb{R}_{\ge 0}$. Each commodity $i \in [k]$ has a strategy set $\mathcal{S}_i \subseteq 2^E$ and demand $r_i \in \mathbb{R}_{>0}$. 
Note that in general the strategy set $\mathcal{S}_i$ of player $i$ is defined by arbitrary resource subsets. If each strategy $P \in \mathcal{S}_i$ corresponds to an $s_i, t_i$-path in a given directed graph, then the corresponding game is called a \emph{network} congestion game.\footnote{If a network congestion game with a single commodity is considered (i.e., $k = 1$), we omit the commodity index for ease of notation.}
We slightly abuse terminology and use the term \emph{path} also to refer to a strategy $P \in \mathcal{S}_i$ of player $i$ (which does not necessarily correspond to a path in a graph); no confusion shall arise.
We denote by $\mathcal{S} = \cup_i \mathcal{S}_i$ the set of all paths.

An outcome of the game is a (feasible) flow $f^i: \mathcal{S}_i \rightarrow \R_{\geq 0}$ satisfying $\sum_{P \in \mathcal{S}_i} f_P^i = r_i$ for every $i \in [k]$. We use $\mathcal{F}(\mathcal{S})$ to denote the set of all feasible flows $f = (f^1,\dots,f^k)$. 
Given a flow $f = (f^i)_{i \in [k]} \in \mathcal{F}(\mathcal{S})$, we use $f^i_e$ to denote the total flow on resource $e \in E$ of commodity $i \in [k]$, i.e., $f_e^i = \sum_{P \in \mathcal{S}_i : e \in P} f_P^i$. The total flow on edge $e \in E$ is defined as $f_e = \sum_{i \in [k]} f_e^i$. 

The latency of a path $P \in \mathcal{S}$ with respect to $f$ is defined as $l_P(f) := \sum_{e \in P} l_e(f_e)$. The cost of commodity $i$ with respect to $f$ is $C_i(f) =  \sum_{P \in \mathcal{S}_i} f_P l_P(f)$. The \textit{social cost} $C(f)$ of a flow $f$ is given by its total average latency, i.e., $C(f) = \sum_{i \in [k]} C_i(f) = \sum_{e \in E} f_e l_e(f_e)$. 
A flow that minimizes $C(\cdot)$ is called \textit{(socially) optimal}. 


If the population is heterogenous, then each commodity $i \in [k]$ is further partitioned in $h_i$ \emph{sensitivity classes}, where class $j \in [h_i]$ has demand $r_{ij}$ such that $r_i = \sum_{j \in [h_i]} r_{ij}$. 
Given a path $P \in \mathcal{S}_i$, we use $f_{P,j}$ to refer to the amount of flow on path $P$ of sensitivity class $j$ (so that $\sum_{j \in [h_i]} f_{P,j} = f_P$).

\paragraph{Deviated Nash flows.} 


We consider a \emph{bounded deviation model} similar to the one introduced in \cite{Kleer2016}.\footnote{In fact, in \cite{Kleer2016} more general path deviations are introduced; the path deviations considered here correspond to \emph{$(0, \beta)$-path deviations} in \cite{Kleer2016}.}
We use $\delta= (\delta_P)_{P \in \mathcal{S}}$ to denote some arbitrary path deviations, where $\delta_P : \mathcal{F}(\mathcal{S}) \rightarrow \R_{\geq 0}$ for all $P \in \mathcal{S}$.
Let $\beta \geq 0$ be fixed. 
Define the set of \emph{$\beta$-bounded path deviations} as
$
\Delta(\beta) = \{ (\delta_P)_{P \in \mathcal{S}} \ | \ 0 \leq \delta_P(f) \leq \beta l_P(f) \text{ for all } f \in \mathcal{F}(\mathcal{S}) \}.
$

Every commodity $i \in [k]$ and sensitivity class $j \in [h_i]$ has a non-negative sensitivity $\gamma_{ij}$ with respect to the path deviations. The population is \emph{homogeneous} if $\gamma_{ij} = \gamma$ for all $i \in [k]$, $j \in [h_i]$ and some $\gamma \ge 0$; otherwise, it is \emph{heterogeneous}. 
Define the \emph{deviated latency} of a path $P \in \mathcal{S}_i$ for sensitivity class $j \in [h_i]$ as $q_{P}^j(f) = l_P(f) + \gamma_{ij} \delta_P(f)$.


We say that a flow $f$ is a \emph{$\beta$-\perturbed Nash flow} if there exist some $\beta$-bounded path deviations $\delta \in \Delta(\beta)$ such that
\begin{equation}
\forall i \in [k], \forall j \in [h_i], \forall P \in \mathcal{S}_i, f_{P,j} > 0: \qquad q_{P}^j(f) \leq q_{P'}^j(f)  \ \ \forall P' \in \mathcal{S}_i.
\label{eq:nash}
\end{equation}
We define the \emph{$\beta$-\deviation ratio} $\bDR(\mathcal{I})$ as the maximum ratio $C(f^\beta)/C(f^0)$ of an $\beta$-\perturbed Nash flow $f^\beta$ and an original Nash flow $f^0$. 
Intuitively, the deviation ratio measures the worst-case deterioration in social cost as a result of (bounded) deviations in the path latencies.
Note that here the comparison is done with respect to an \emph{unaltered} Nash flow to measure the impact of these deviations.

The set $\Delta(\beta)$ can also be restricted to path deviations which are defined as a function of edge deviations along that path. Suppose every edge $e \in E$ has a deviation $\delta_e : \R_{\geq 0} \rightarrow \R_{\geq 0}$ satisfying $0 \leq \delta_e(x) \leq \beta l_e(x)$ for all $x \geq 0$. For example, feasible path deviations can then be defined by the $L_1$-norm objective $\delta_P(f) = \sum_{e \in P} \delta_e(x)$ (as in \cite{Kleer2016,Nikolova2015}) or the $L_2$-norm objective $\delta_P(f) = \sqrt{\sum_{e \in P} \delta_e(x)^2)}$ (as in \cite{Nikolova2015,Lianeas2016}).
The \emph{Price of Risk-Aversion} introduced by Nikolova and Stier-Moses \cite{Nikolova2015} is technically the same ratio as the \deviation ratio for the $L_1$- and $L_2$-norm (see \cite{Kleer2016} for details).

\paragraph{Approximate Nash flows.}


We introduce the notion of an approximate Nash flow. Also here, each commodity $i \in [k]$ and sensitivity class $j \in [h_i]$ has a non-negative sensitivity $\epsilon_{ij}$.
We say that the population is \emph{homogeneous} if $\epsilon_{ij} = \epsilon$ for all $i \in [k]$, $j \in [h_i]$ and some $\epsilon \ge 0$; otherwise, it is \emph{heterogeneous}.

A flow $f$ is an \emph{$\epsilon$-approximate Nash flow} with respect to sensitivities $\epsilon = (\epsilon_{ij})_{i \in [k], j \in [h_i]}$ if
\begin{equation}
\forall i \in [k], \ \forall j \in [h_i], \ \forall P \in \mathcal{S}_i, f_{P,j} > 0: \qquad l_{P}(f) \leq (1+\epsilon_{ij}) l_{P'}(f)  \ \ \forall P' \in \mathcal{S}_i
\label{def:epsnash}
\end{equation}
Note that a $0$-approximate Nash flow is simply a Nash flow. We define the \textit{$\epsilon$-\stability ratio} $\eSR(\mathcal{I})$ as the maximum ratio $C(f^\epsilon)/C(f^0)$ of an $\epsilon$-approximate Nash flow $f^\epsilon$ and an original Nash flow $f^0$.
 
Some of the proofs are missing in the main text below and can be found in the appendix.

\section{Heterogeneous populations}\label{sec:het}

We first elaborate on the relation between \perturbed and approximate Nash flows for general congestion games with heterogeneous populations. 


\begin{proposition}\label{prop:inclusion}
Let $\mathcal{I}$ be a congestion game with heterogeneous players. 
If $f$ is a $\beta$-\perturbed Nash flow for $\mathcal{I}$, then $f$ is an $\epsilon$-approximate Nash flow for $\mathcal{I}$ with $\epsilon_{ij} = \beta \gamma_{ij}$ for all $i \in [k]$ and $j \in [h_i]$ (for the same demand distribution $r$).
\end{proposition}


\paragraph{Discrete sensitivity distributions.}

Subsequently, we show that the reverse of Proposition \ref{prop:inclusion} does not hold. 
We do this by providing tight bounds on the $\beta$-\deviation ratio and the $\epsilon$-stability ratio for instances on (single-commodity) series-parallel graphs and arbitrary discrete sensitivity distributions.


\begin{theorem}\label{thm:hetero}
Let $\mathcal{I}$ be a single-commodity network congestion game on a series-parallel graph with heterogeneous players, demand distribution $r = (r_i)_{i \in [h]}$ normalized to $1$, i.e.,  $\sum_{j \in [h]} r_i = 1$, and sensitivity distribution $\gamma = (\gamma_i)_{i \in [h]}$, with $\gamma_1 < \gamma_2 < \dots < \gamma_h$. 
Let $\beta \ge 0$ be fixed and define $\epsilon = (\beta \gamma_i)_{i \in [h]}$. Then the $\epsilon$-stability ratio and the $\beta$-\deviation ratio are bounded by: 
\begin{align}
\eSR(\mathcal{I})  \le 1 + \beta \sum_{j = 1}^h r_j\gamma_j \quad\text{and}\quad
\bDR(\mathcal{I})  \le 1 + \beta \cdot \max_{j \in [h]} \bigg\{ \gamma_j \bigg( \sum_{p = j}^h r_p \bigg) \bigg\}.
\label{eq:comparison_epsilon}
\end{align}
Further, both bounds are tight for all distributions $r$ and $\gamma$.
\end{theorem}
It is not hard to see that the bound on the $\beta$-\deviation ratio is always smaller than the bound on the $\epsilon$-stability ratio.\footnote{This follows from Markov's inequality: for a random variable $Y$, $P(Y \geq t) \leq E(Y)/t$.}
Our bound on the $\beta$-\deviation ratio also yields tight bounds on the \emph{Price of Risk-Aversion} \cite{Nikolova2015} for series-parallel graphs and arbitrary heterogeneous risk-averse populations, both for the $L_1$-norm and $L_2$-norm objective.\footnote{Observe that we show tightness of the bound on parallel arcs, in which case these objectives coincide.}

We need the following technical lemma for the proof of the $\beta$-deviation ratio.

\begin{lemma}\label{lem:max_seq}
Let $0 \leq \tau_{k-1} \leq \dots \leq \tau_1 \leq \tau_0$ and $c_i \geq 0$ for $i = 1,\dots,k$ be given. We have
$
c_1 \tau_0 + \sum_{i = 1}^{k-1} (c_{i+1} - c_i)\tau_i \leq \tau_0 \cdot \max_{i=1,\dots,k} \{c_i\}.
$
\end{lemma}


\begin{proof}[Theorem~\ref{thm:hetero}, $\beta$-\deviation ratio]
Let $x = f^\beta$ be a $\beta$-\perturbed Nash flow with path deviations $(\delta_P)_{P \in \mathcal{S}} \in \Delta(\beta)$ and let $z = f^0$ be an original Nash flow.
Let $X = \{a \in A : x_a > z_a\}$ and $Z = \{a \in A : z_a \geq x_a \text{ and } z_a > 0\}$ (arcs with $x_a = z_a = 0$ may be removed without loss of generality).

In order to analyze the ratio $C(x)/C(z)$ we first argue that we can assume without loss of generality that the latency function $l_a(y)$ is constant for values $y \geq x_a$ for all arcs $a \in Z$. To see this, note that we can replace the function $l_a(\cdot)$ with the function $\hat{l}_a$ defined by  $\hat{l}_a(y) = l_a(x_a)$ for all $y \geq x_a$ and $\hat{l}_a(y) = l_a(y)$ for $y \leq x_a$. In particular, this implies that the flow $x$ is still a $\beta$-\perturbed Nash flow for the same path deviations as before. This holds since for any path $P$ the latency $l_P(x)$ remains unchanged if we replace the function $l_a$ by $\hat{l}_a$. 

By definition of arcs in $Z$, we have $x_a \leq z_a$ and therefore $\hat{l}_a(z_a) = l_a(x_a) \leq l_a(z_a)$. Let $z'$ be an original Nash flow for the instance with $l_a$ replaced by $\hat{l}_a$. Then we have $C(z') \leq C(z)$ using the fact that series-parallel graphs are immune to the Braess paradox, see Milchtaich \cite[Lemma 4]{Milchtaich2006}. Note that, in particular, we find
$C(x)/C(z) \leq C(x)/C(z')$. By repeating this argument, we may without loss of generality assume that all latency functions $l_a$ are constant between $x_a$ and $z_a$ for $a \in Z$. Afterwards, we can even replace the function $\hat{l}_a$ by a function that has the constant value of $l_a(x_a)$ everywhere. 
In the remainder of the proof, we will denote $P_j$ as a flow-carrying arc for sensitivity class $j \in [h]$ that maximizes the path latency amongst all flow-carrying path for sensitivity class $j \in [h]$, i.e., 
$
P_j = \text{argmax}_{P \in \mathcal{P} : x_{P,j} > 0} \{ l_P(x) \}.
$
Moreover, there also exists a path $P_0$ with the property that $z_a \geq x_a$ and $z_a > 0$ for all arcs $a \in P_0$ (see, e.g., Lemma 2 \cite{Milchtaich2006}).

For fixed $a < b \in \{1,\dots,h\}$, the Nash conditions imply that (these steps are of a similar nature as Lemma 1 \cite{Fleischer2005})
\begin{align*}
l_{P_a}(x) + \gamma_{a}\cdot \delta_{P_a}(x) &\leq l_{P_b}(x) + \gamma_{a} \cdot \delta_{P_b}(x)  \\
l_{P_b}(x) + \gamma_{b}\cdot \delta_{P_b}(x) &\leq l_{P_a}(x) + \gamma_{b} \cdot \delta_{P_a}(x).  
\end{align*}
Adding up these inequalities implies that $(\gamma_b - \gamma_a) \delta_{P_b}(x) \leq (\gamma_b - \gamma_a) \delta_{P_a}(x)$, which in turn yields that $\delta_{P_b}(x) \leq \delta_{P_a}(x)$ (using that $\gamma_a < \gamma_b$ if $a < b$). Furthermore, we also have
\begin{equation}\label{eq:nash1}
l_{P_1}(x) + \gamma_1 \delta_{P_1}(x) \leq l_{P_0}(x) + \gamma_1 \delta_{P_0}(x),
\end{equation}
and $l_{P_0}(x) = l_{P_0}(z) \leq l_{P_1}(z) \leq l_{P_1}(x)$, which can be seen as follows. The equality follows from the fact that $l_a$ is constant for all $a \in Z$ and, by choice, $P_0$ only consists of arcs in $Z$. The first inequality follows from the Nash conditions of the original Nash flow $z$, since there exists a flow-decomposition in which the path $P_0$ is used (since the flow on all arcs of $P_0$ is strictly positive in $z$). The second inequality follows from the fact that 
$$
\sum_{e \in P_1} l_e(z_e) = \sum_{e \in P_1 \cap X} l_e(z_e)  + \sum_{e \in P_1 \cap Z} l_e(z_e) \leq \sum_{e \in P_1 \cap X} l_e(x_e)  + \sum_{e \in P_1 \cap Z} l_e(x_e)
$$
using that $z_e \leq x_e$ for $e \in X$ and the fact that latency functions for $e \in Z$ are constant. In particular, we find that $l_{P_0}(x) \leq l_{P_1}(x)$. Adding this inequality to  (\ref{eq:nash1}), we obtain $\gamma_1\delta_{P_1}(x) \leq \gamma _1 \delta_{P_0}(x)$ and therefore $\delta_{P_1}(x) \leq \delta_{P_0}(x)$. Thus
$
\delta_{P_h}(x) \leq \delta_{P_{h-1}}(x) \leq \dots \leq \delta_{P_1}(x) \leq \delta_{P_0}(x).
$
Moreover, by using induction (see appendix) it can be shown that
\begin{equation}\label{eq:induction}
l_{P_j}(x) \leq l_{P_0}(x) + \gamma_1 \delta_{P_0}(x) + \bigg[\sum_{g = 1}^{j-1} (\gamma_{g+1} - \gamma_g)\delta_{P_g}(x)\bigg] - \gamma_j \delta_{P_j}(x).
\end{equation}
Using (\ref{eq:induction}), we then have
\begin{eqnarray}
C(x) &\leq & \sum_{j=1}^h r_j l_{P_j}(x) \ \ \ \ \text{(by choice of the paths $P_j$)} \nonumber \\
&\leq& \sum_{j=1}^h r_j \left(l_{P_0}(x) + \gamma_1 \delta_{P_0}(x) + \bigg[\sum_{g = 1}^{j-1} (\gamma_{g+1} - \gamma_g)\delta_{P_g}(x)\bigg] - \gamma_j \delta_{P_j}(x) \right) \nonumber \\
&=& l_{P_0}(x) + \gamma_1 \delta_{P_0}(x) + \sum_{j=1}^{h} (r_{j+1}+\dots+ r_h)(\gamma_{j+1} - \gamma_j)\delta_{P_j}(x) - r_j\gamma_j \delta_{P_j}(x) \nonumber  \\
&\leq & l_{P_0}(x) + \gamma_1 \delta_{P_0}(x) \nonumber \\
& & +\sum_{j=1}^{h-1} \bigg[(r_{j+1}+\dots+ r_h)\gamma_{j+1} - (r_j+r_{j+1}+\dots+ r_h)\gamma_j\bigg]\delta_{P_j}(x) \nonumber
\end{eqnarray}
In the last inequality, we leave out the last negative term $-r_h\gamma_h\delta_{P_h}(x)$.
Note that $\gamma_1 = (r_1 + \dots + r_h)\gamma_1$ since we have normalized the demand to $1$. 
We can then apply Lemma \ref{lem:max_seq} with $\tau_i = \delta_{P_i}(x)$ for $i = 0,\dots, h-1$ and 
$
c_i = \gamma_i \cdot \sum_{p = i}^h r_p 
$
for $i = 1,\dots, k$. 
Continuing the estimate, we get
$$
C(x) \leq l_{P_0}(x) + \max_{j \in [h]} \bigg\{ \gamma_j \cdot \sum_{p = j}^h r_p  \bigg\} \cdot \delta_{P_0}(x)  \leq \bigg[1 + \beta \cdot \max_{j \in [h]} \bigg\{ \gamma_j \bigg( \sum_{p = j}^h r_p \bigg) \bigg\}\bigg] C(z)
$$
where for the second inequality we use that $\delta_{P_0}(x) \leq \beta l_{P_0}(x)$, which holds by definition, and $l_{P_0}(x) = l_{P_0}(z) = C(z)$, which holds because $z$ is an original Nash flow and all arcs in $P_0$ have strictly positive flow in $z$ (and because of the fact that that all arcs in $P_0$ have a constant latency functions). 

To prove tightness, fix $j \in [h]$ and consider the following instance on two arcs. We take $(l_1(y),\delta_1(y)) = (1,\beta)$ and $(l_2(y),\delta_2(y))$ with $\delta_2(y) = 0$ and $l_2(y)$ a strictly increasing function satisfying $l_2(0) = 1 + \epsilon$ and $l_2(r_j+r_{j+1}+\dots+r_h) = 1 + \gamma_j \beta$, where $\epsilon < \gamma_j \beta$. The (unique) original Nash flow is given by $z = (z_1,z_2) = (1,0)$ with $C(z) = 1$. The (unique) $\beta$-\perturbed Nash flow $x$ is given by $x = (x_1,x_2) = (r_1 + r_2 + \dots + r_{j-1}, r_j+r_{j+1}+\dots+r_h)$ with $C(x) = 1 + \beta \cdot \gamma_j (r_j + \dots + r_h)$. Since this construction holds for all $j \in [h]$, we find the desired lower bound.\qed

\end{proof}

\paragraph{Continuous sensitivity distributions.}

We obtain a similar result for more general (not necessarily discrete) sensitivity distributions. That is, we are given a Lebesgue integrable \emph{sensitivity density function} $\psi : \R_{\geq 0} \rightarrow \R_{\geq 0}$ over the total demand. Since we can normalize the demand to $1$, we have the condition that $\int_0^\infty \psi(y) dy = 1$.  We then find the following natural generalizations of our upper bounds:
\begin{enumerate}
\item $\eSR(\mathcal{I}) \le 1 + \beta \int_{0}^\infty y\cdot \psi(y)dy$, and
\item $\bDR(\mathcal{I}) \le 1 + \beta \cdot \sup_{t \in \R_{\geq 0}} \big\{ t \cdot \int_t^\infty \psi(y) dy \big\}$. 
\end{enumerate}
These bounds are both asymptotically tight for all distributions. Details are given in Corollary \ref{cor:het} in the appendix.

\section{Homogeneous population}\label{sec:approx}

The reverse of Proposition~\ref{prop:inclusion} also holds for homogeneous players in single-commodity instances. 
As a consequence, the set of $\beta$-\perturbed Nash flows and the set of $\epsilon$-approximate Nash flows with $\epsilon = \beta \gamma$ coincide in this case.

Recall that for homogeneous players we have $\gamma_{ij} = \gamma$ for all $i \in [k]$, $j \in [h_i]$ and some $\gamma \ge 0$. 

\begin{proposition}\label{prop:equiv}
Let $\mathcal{I}$ be a single-commodity congestion game with homogeneous players. $f$ is an $\epsilon$-approximate Nash flow for $\mathcal{I}$ if and only if $f$ is a $\beta$-\perturbed Nash flow for $\mathcal{I}$ with $\epsilon = \beta\gamma$. 
\end{proposition}


\noindent \emph{Upper bound on the stability ratio.}
Our main result in this section is an upper bound on the $\epsilon$-stability ratio. Given the above equivalence, this bound also applies to the $\beta$-\deviation ratio with $\epsilon = \beta\gamma$.

The following concept of alternating paths is crucial. For single-commodity instances an alternating path always exists (see, e.g., \cite{Nikolova2015}).


\begin{definition}[Alternating path \cite{Lin2011,Nikolova2015}]
Let $\mathcal{I}$ be a single-commodity network congestion game and let $x$ and $z$ be feasible flows. We partition the edges $E = X \cup Z$ such that 
$Z = \{a \in E : z_a \geq x_a \text{ and } z_a > 0\}$ and $X = \{a \in E : z_a < x_a \text{ or } z_a = x_a = 0\}$. We say that $\pi$
is an alternating $s,t$-path if the arcs in $\pi \cap Z$ are oriented in the direction of $t$, and the arcs in $\pi \cap X$ are oriented in the direction of $s$.
We call the number of backward arcs on $\pi$ the \emph{backward length} of $\pi$ and refer to it by $q(\pi) = |\pi \cap X|$.
\label{def:alt_path}
\end{definition}

\begin{theorem}\label{thm:approx_single}
Let  $\mathcal{I}$ be a single-commodity network congestion game. Let $\epsilon \ge 0$ be fixed and consider an arbitrary alternating path $\pi$ with backward length $q = q(\pi)$. If $\epsilon < 1/q$, then the $\epsilon$-stability ratio is bounded by 
$$
\eSR(\mathcal{I}) \le \frac{1 + \epsilon}{1 - \epsilon \cdot q} \leq \frac{1 + \epsilon}{1 - \epsilon \cdot n}.
$$
\end{theorem}

Note that the restriction on $\epsilon$ stated in the theorem always holds if 
$\epsilon < 1/n$. In particular, for $\epsilon \ll 1/n$ we roughly get $\eSR(\mathcal{I}) \leq 1 + \epsilon n$. 
The proof of Theorem~\ref{thm:approx_single} is inspired by a technique of Nikolova and Stier-Moses \cite{Nikolova2015}, but technically more involved.

\begin{proof}

Let $x = f^\epsilon$ be an $\epsilon$-approximate Nash flow and let $z = f^0$ an original Nash flow. Let $\pi = Z_1X_1Z_2X_2\dots Z_{\eta-1}X_{\eta-1}Z_{\eta}$ be an alternating path for $x$ and $z$, where $Z_i$ and $X_i$ are maximal sections consisting of consecutive arcs, respectively, in $Z$ and $X$ (i.e., $Z_i \subseteq Z$ and $X_i \subseteq X$ for all $i$).
Furthermore, we let $q_i = |X_i|$ and write $X_i =  (X_{iq_i},\dots,X_{i2},X_{i1})$, where $X_{ij}$ are the arcs in the section $X_i$. 
By definition, for every arc $X_{ij}$ there exists a  path $C_{ij}X_{ij}D_{ij}$ that is flow-carrying for $x$.\footnote{Note that for a Nash flow one can assume that there is a flow-carrying path traversing all arcs $X_{iq_i},\dots,X_{i1}$; but this cannot be done for an approximate Nash flow.}

For convenience, we define $C_{01} = D_{\eta,0} = \emptyset$. Furthermore, we denote $P^{\max}$ as a path maximizing $l_P(x)$ over all paths $P \in \mathcal{S}$. For convenience, we will abuse notation, and write $Q = Q(x) = \sum_{a \in Q} l_{a}(x)$ for $Q \subseteq E$.

Note that for all $i,j$:
\begin{equation}\label{eq:pmax_upperbound}
C_{ij}(x) + X_{ij}(x) + D_{ij}(x) \leq P^{\max}(x).
\end{equation}

\begin{figure}[t!] \centering

\scalebox{0.8}{
\begin{tikzpicture}[
  ->,
  >=stealth',
  shorten >=0.5pt,
  auto,
  semithick,
  every state/.style={circle,minimum size=1pt},
]
r = 1
\begin{scope}

  \node[state]  (s)         					 {};
  \node[state]  (v4) [right=3cm of s]  			 {};
  \node[state]  (w4) [above=0.75cm of v4] {};
  \node[state]  (w3) [right=3cm of v4] 				 {};     
  \node[state]  (v3) [below=0.75cm of w3] 				 {};
  \node[state]  (v3b) [below=0.75cm of v3] 				 {};
  \node[state]  (t)  [right=3cm of w3] 				 {};
  
\path[every node/.style={sloped,anchor=south,auto=false}]
(s) edge[line width=1.5pt, bend left=15] 	node {$Z_1$} (w4)            
(s) edge						node {$C_{11}$} (v4)
(s) edge[bend right=10]        				node {$C_{21}$} (v3)
(s) edge[bend right=15]            			node {$C_{22}$} (v3b)
(v4) edge[line width=1.5pt]             	node {$X_{11}$} (w4)
(v4) edge[line width=1.5pt]         		node {$Z_2$} (w3)
(v3) edge[line width=1.5pt]               	node {$X_{21}$} (w3)
(v3b) edge[line width=1.5pt]               	node {$X_{22}$} (v3)
(w4) edge[bend left=10]node {$D_{11}$} (t)            
(w3) edgenode {$D_{21}$} (t)
(v3) edge[bend right=10] node {$D_{22}$} (t)
(v3b) edge[line width=1.5pt, bend right=15] node {$Z_3$} (t);
        
\end{scope}
\end{tikzpicture}}
\label{fig:braess}
\caption{Sketch of the situation in the proof of Theorem \ref{thm:approx_single} with $q_1 = 1$ and $q_2 = 2$.}
\end{figure}
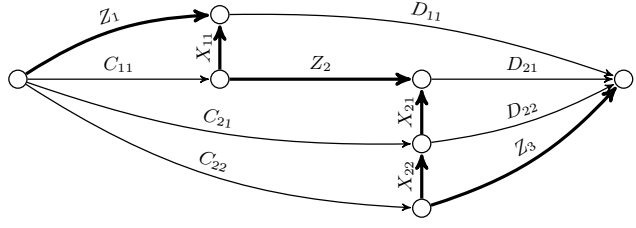
Fix some $i \in \{1,\dots,\eta-1\}$. Then we have $C_{i1} + X_{i1} + D_{i1} \leq (1+\epsilon)(C_{i-1,q_{i-1}} + Z_i + D_{i1})$ by definition of an $\epsilon$-approximate Nash flow. This implies that (leaving out $D_{i1}$ on both sides)
\begin{eqnarray}
C_{i1} + X_{i1} &\leq &(1+\epsilon)Z_i + C_{i-1,q_{i-1}} + \epsilon(C_{i-1,q_{i-1}} + D_{i1}). \nonumber 
\end{eqnarray}
Furthermore, for all $j \in \{2,\dots,q_i\}$, we have $C_{ij} + X_{ij} + D_{ij} \leq (1+\epsilon)(C_{i,j-1} + D_{ij})$ which implies (again leaving out $D_{ij}$ on both sides)
$$
C_{ij} + X_{ij} \leq C_{i,j-1} + \epsilon(C_{i,j-1} + D_{ij}).
$$
Adding up these inequalities for $j \in\{1,\dots,q_i\}$ and 
subtracting $\sum_{j=1}^{q_i-1} C_{ij}$ from both sides, we obtain for all $i \in \{1,\dots,\eta-1\}$
\begin{equation}\label{eq:alt_general_i}
C_{i,q_i} + \sum_{j = 1}^{q_i} X_{ij} \leq C_{i-1,q_{i-1}} + (1+\epsilon)Z_i + \epsilon \bigg( \sum_{j = 1}^{q_i} D_{ij} + C_{i-1,q_{i-1}} + \sum_{j=1}^{q_i-1}C_{ij}\bigg).
\end{equation}
Moreover, we also have
\begin{equation}\label{eq:alt_pmax}
P^{\max} \leq (1+\epsilon)(C_{\eta-1,\eta-1} + Z_{\eta}) =  C_{\eta-1,\eta-1} + (1+\epsilon)Z_{\eta} + \epsilon C_{\eta-1,\eta-1}.
\end{equation}
Adding up the inequalities in (\ref{eq:alt_general_i}) for all $i \in \{1,\dots,\eta-1\}$, and the inequality in (\ref{eq:alt_pmax}), we obtain
$$
P^{\max} + \sum_{i=1}^{\eta-1} C_{i,q_i} + \sum_{i = 1}^{\eta-1} \sum_{j=1}^{q_i} X_{ij} \leq \sum_{i=1}^{\eta-1} C_{i,q_i} + (1+\epsilon) \sum_{i = 1}^{\eta} Z_i + \epsilon \bigg(\sum_{i = 1}^{\eta-1} \sum_{j=1}^{q_i} C_{ij} + D_{ij} \bigg)
$$
which simplifies to
\begin{equation}\label{eq:cum}
P^{\max} + \sum_{i = 1}^{\eta-1} \sum_{j=1}^{q_i} X_{ij} \leq  (1+\epsilon) \sum_{i = 1}^{\eta} Z_i + \epsilon \bigg(\sum_{i = 1}^{\eta-1} \sum_{j=1}^{q_i} C_{ij} + D_{ij} \bigg).
\end{equation}
Using (\ref{eq:pmax_upperbound}), we obtain
$$
\sum_{i = 1}^{\eta-1} \sum_{j=1}^{q_i} C_{ij} + D_{ij} \leq \sum_{i = 1}^{\eta-1} \sum_{j=1}^{q_i} P^{\max} - X_{ij} = \bigg(\sum_{i=1}^{\eta-1} q_i\bigg)P^{\max} - \sum_{i = 1}^{\eta-1} \sum_{j=1}^{q_i} X_{ij}.
$$
Combining this with (\ref{eq:cum}), and rearranging some terms, we get
\begin{align*}\label{eq:cum_final}
(1 - \epsilon \cdot q)P^{\max} &\leq  (1+\epsilon) \bigg[ \sum_{i = 1}^{\eta} Z_i -  \sum_{i = 1}^{\eta-1} \sum_{j=1}^{q_i} X_{ij}\bigg] 
= (1+\epsilon) \bigg[ \sum_{e \in Z \cap \pi} l_e(x_e) -  \sum_{e \in X \cap \pi} l_e(x_e)\bigg]
\end{align*}
where $q = q(\pi) = \sum_{i = 1}^{\eta-1} q_i$ is the backward length of $\pi$. 

Similarly (see also \cite[Lemma 4.5]{Nikolova2015}), it can be shown that
\begin{equation}\label{eq:z}
l_Q(z) \geq \sum_{e \in Z \cap \pi} l_e(z_e) - \sum_{e \in X \cap \pi} l_e(z_e)
\end{equation}
for any path $Q$ with $z_Q > 0$ (these all have the same latency, since $z$ is an original Nash flow). Using a similar argument as in \cite[Theorem 4.6]{Nikolova2015}, we obtain
\begin{eqnarray}
(1 - \epsilon \cdot q)l_{P^{\max}}(x) &\leq & (1+\epsilon) \bigg[ \sum_{e \in Z \cap \pi} l_e(x_e) -  \sum_{e \in X \cap \pi} l_e(x_e)\bigg] \nonumber \\
& \leq & (1+\epsilon) \bigg[ \sum_{e \in Z \cap \pi} l_e(z_e) -  \sum_{e \in X \cap \pi} l_e(z_e)\bigg] 
\nonumber 
\le 
(1+\epsilon)l_Q(z).
\end{eqnarray}
By multiplying both sides with the demand $r$, we obtain
$
(1 - \epsilon \cdot q) C(x) \leq (1 - \epsilon \cdot q) r \cdot l_{P^{\max}}(x) \leq (1+\epsilon)r \cdot l_Q(z) = (1+\epsilon) C(z)
$
for $\epsilon < 1/q$, which proves the claim.
\qed
\end{proof}



\paragraph{Tight bound on the stability ratio.}

In this section, we consider instances for which all backward sections of the alternating path $\pi$ consist of a single arc., i.e., $q_i = 1$ for all $i = 1,\dots,\eta-1$. We then have $q = \sum_{i=1}^{\eta-1} q_i \leq \lfloor n/2 \rfloor - 1$ since every arc in $X$ must be followed directly by an arc in $Z$ (and we can assume w.l.o.g. that the first and last arc are contained in $Z$). 
By Theorem \ref{thm:approx_single}, we obtain
$
\eSR(\mathcal{I}) \leq (1 + \epsilon)/(1 - \epsilon \cdot (\lfloor n/2 \rfloor - 1))
$
for all $\epsilon < 1/(\lfloor n/2 \rfloor - 1)$. 
We show that this bound is tight. Further, we show that there exist instances for which $\eSR(\mathcal{I})$ is unbounded for $\epsilon \geq  1/(\lfloor n/2 \rfloor - 1)$.
This completely settles the case of $q_i = 1$ for all $i$.

Our construction is based on the \emph{generalized Braess graph} \cite{Roughgarden2006}. 
 By construction, alternating paths for these graphs satisfy $q_i = 1$ for all $i$ (see Figure \ref{fig:braess_5} in the appendix for an example and a formal definition of these graphs).

\begin{theorem}\label{thm:braess_tight}
Let $n = 2m$ be fixed and let $\mathcal{B}^m$ be the set of all instances on the generalized Braess graph with $n$ nodes. Then
$$
\sup_{\mathcal{I} \in \mathcal{B}^m} \eSR(\mathcal{I}) 
= 
\begin{cases}
\frac{1 + \epsilon}{1 - \epsilon \cdot (\lfloor n/2 \rfloor - 1)} & \text{if $\epsilon < \frac{1}{\lfloor n/2 \rfloor - 1}$,} \\
\infty & \text{otherwise.}
\end{cases}
$$
\end{theorem}

\paragraph{Non-symmetric matroid congestion games.}\label{sec:matroid}

In the previous sections, we considered (symmetric) network congestion games only. It is interesting to consider other combinatorial strategy sets as well. In this section we make a first step in this direction by focusing on the bases of matroids as strategies. 

A matroid congestion game is given by $\mathcal{J} = (E,(l_e)_{e \in E},(\mathcal{S}_i)_{i \in [k]},(r_i)_{i \in [k]})$, and matroids $\mathcal{M}_i = (E,\mathcal{I}_i)$ over the ground set $E$ for every $i \in [k]$.\footnote{A matroid over $E$ is given by a collection $\mathcal{I} \subseteq 2^E$ of subsets of $E$ (called \emph{independent sets}). The pair $\mathcal{M} = (E,\mathcal{I})$ is a \emph{matroid} if the following three properties hold: i) $\emptyset \in \mathcal{I}$; ii) If $A \in \mathcal{I}$ and $B \subseteq A$, then $B \in \mathcal{I}$. iii) If $A,B \in \mathcal{I}$ and $|A| > |B|$, then there exists an $a \in A \setminus B$ such that $B + a \in \mathcal{I}$.}  
The strategy set $\mathcal{S}_i$ consists of the \emph{bases} of the matroid $\mathcal{M}_i$, which are the independent sets of maximum size, e.g., spanning trees in an undirected graph. We refer the reader to Schrijver \cite{schrijver2003} for an extensive overview of matroid theory.


As for network congestion games, it can be shown that in general the $\epsilon$-stability ratio can be unbounded (see Theorem \ref{thm:matroid_unbounded} in the appendix); this also holds for general path deviations because the proof of Proposition~\ref{prop:equiv} in the appendix holds for arbitrary strategy sets.
However, if we consider path deviations induced by the sum of edge deviations (as in \cite{Nikolova2015,Kleer2016}), then we can obtain a more positive result for general matroids.

%

%


Recall that for every resource $e \in E$ we have a deviation function $\delta_e : \R_{\geq 0} \rightarrow \R_{\geq 0}$ satisfying $0 \leq \delta_e(x) \leq \beta l_e(x)$ for all $x \geq 0$. 
The deviation of a basis $B$
is then given by $\delta_B(f) = \sum_{e \in B} \delta_e(f_e)$. 


\begin{theorem}\label{thm:matroid}
Let $\mathcal{J} = (E,(l_e)_{e \in E},(\mathcal{S}_i)_{i \in [k]},(r_i)_{i \in [k]})$ be a matroid congestion game with homogeneous players. Let $\beta \ge 0$ be fixed and consider $\beta$-bounded basis deviations as defined above. 
Then the $\beta$-\deviation ratio is upper bounded by $\bDR(\mathcal{J}) \le 1 + \beta$. Further, this bound is tight already for $1$-uniform matroid congestion games. 
\end{theorem} 

\noindent \textbf{Acknowledgements.} We thank the anonymous referees for their very useful comments, and one reviewer for pointing us to Lemma 1 \cite{Fleischer2005}.

 \bibliographystyle{abbrv}
\bibliography{references}

\begin{thebibliography}{10}

\bibitem{Chen2014}
P.-A. Chen, B.~D. Keijzer, D.~Kempe, and G.~Sch\"{a}fer.
\newblock Altruism and its impact on the price of anarchy.
\newblock {\em ACM Trans. Econ. Comput.}, 2(4):17:1--17:45, Oct. 2014.

\bibitem{Christodoulou2011}
G.~Christodoulou, E.~Koutsoupias, and P.~G. Spirakis.
\newblock On the performance of approximate equilibria in congestion games.
\newblock {\em Algorithmica}, 61(1):116--140, 2011.

\bibitem{Englert2008}
M.~Englert, T.~Franke, and L.~Olbrich.
\newblock {\em Sensitivity of Wardrop Equilibria}, pages 158--169.
\newblock Springer Berlin Heidelberg, Berlin, Heidelberg, 2008.

\bibitem{Fleischer2005}
L.~Fleischer.
\newblock Linear tolls suffice: New bounds and algorithms for tolls in single
  source networks.
\newblock {\em Theoretical Computer Science}, 348(2):217 -- 225, 2005.

\bibitem{Fujishige2015}
S.~Fujishige, M.~X. Goemans, T.~Harks, B.~Peis, and R.~Zenklusen.
\newblock Matroids are immune to braess paradox.
\newblock {\em CoRR}, abs/1504.07545, 2015.

\bibitem{Kleer2016}
P.~Kleer and G.~Sch{\"{a}}fer.
\newblock The impact of worst-case deviations in non-atomic network routing
  games.
\newblock {\em CoRR}, abs/1605.01510, 2016.

\bibitem{Koutsoupias1999}
E.~Koutsoupias and C.~Papadimitriou.
\newblock Worst-case equilibria.
\newblock In {\em Proceedings of the 16th Annual Conference on Theoretical
  Aspects of Computer Science}, STACS'99, pages 404--413, Berlin, Heidelberg,
  1999. Springer-Verlag.

\bibitem{Lianeas2016}
T.~Lianeas, E.~Nikolova, and N.~E. Stier-Moses.
\newblock Asymptotically tight bounds for inefficiency in risk-averse selfish
  routing.
\newblock In {\em Proceedings of the Twenty-Fifth International Joint
  Conference on Artificial Intelligence}, IJCAI 2016, pages 338--344, New York,
  NY, USA, 2016.

\bibitem{Lin2011}
H.~Lin, T.~Roughgarden, {\'E}.~Tardos, and A.~Walkover.
\newblock Stronger bounds on braess's paradox and the maximum latency of
  selfish routing.
\newblock {\em SIAM Journal on Discrete Mathematics}, 25(4):1667--1686, 2011.

\bibitem{Meir2014Arxiv}
R.~Meir and D.~C. Parkes.
\newblock Playing the wrong game: Smoothness bounds for congestion games with
  risk averse agents.
\newblock {\em CoRR}, abs/1411.1751, 2017.

\bibitem{Milchtaich2006}
I.~Milchtaich.
\newblock Network topology and the efficiency of equilibrium.
\newblock {\em Games and Economic Behavior}, 57(2):321 -- 346, 2006.

\bibitem{Nikolova2015}
E.~Nikolova and N.~E. Stier-Moses.
\newblock The burden of risk aversion in mean-risk selfish routing.
\newblock In {\em Proceedings of the Sixteenth ACM Conference on Economics and
  Computation}, EC '15, pages 489--506, New York, NY, USA, 2015. ACM.

\bibitem{Roughgarden2006}
T.~Roughgarden.
\newblock On the severity of braess's paradox: Designing networks for selfish
  users is hard.
\newblock {\em J. Comput. Syst. Sci.}, 72(5):922--953, Aug. 2006.

\bibitem{Roughgarden2015}
T.~Roughgarden.
\newblock Intrinsic robustness of the price of anarchy.
\newblock {\em J. ACM}, 62(5):32, 2015.

\bibitem{schrijver2003}
A.~Schrijver.
\newblock {\em Combinatorial optimization : polyhedra and efficiency. Vol. {B}.
  , Matroids, trees, stable sets. chapters 39-69}.
\newblock Algorithms and combinatorics. Springer-Verlag, 2003.

\bibitem{Wardrop1952}
J.~G. Wardrop.
\newblock Some theoretical aspects of road traffic research.
\newblock {\em Proceedings of the Institution of Civil Engineers}, 1:325--378,
  1952.

\end{thebibliography}
\newpage
\appendix

\newpage
\section{Omitted material of Section \ref{sec:het}}

\begin{rtheorem}{Proposition}{\ref{prop:inclusion}}
Let $\mathcal{I}$ be a congestion game with heterogeneous players. 
If $f$ is a $\beta$-\perturbed Nash flow for $\mathcal{I}$, then $f$ is an $\epsilon$-approximate Nash flow for $\mathcal{I}$ with $\epsilon_{ij} = \beta \gamma_{ij}$ for all $i \in [k]$ and $j \in [h_i]$ (for the same demand distribution $r$).
\end{rtheorem}
\begin{proof}
Let $f$ be a $\beta$-deviated Nash flow for some set of path deviations $\delta \in \Delta(0,\beta)$. For $P \in \mathcal{S}$, we have
$$
l_P(f) \leq l_P + \gamma_{ij}\delta_P(f) \leq l_{P'}(f) + \gamma_{ij} \delta_{P'}(f)\leq (1+ \gamma_{ij} \beta) l_{P'}(f) 
$$
where we use the non-negativity of $\delta_P(f)$ in the first inequality, the Nash condition for $f$ in the second inequality, and the fact that $\delta_{P'}(f) \leq \beta l_{P'}(f)$ in the third inequality. By definition, it now holds that $f$ is also a $\epsilon$-approximate equilibrium (for $\epsilon = \beta \gamma$). \qed
\end{proof}

\begin{rtheorem}{Theorem}{\ref{thm:hetero}}
Let $\mathcal{I}$ be a network congestion game on a series-parallel graph with heterogeneous players, demand distribution $r = (r_i)_{i \in [h]}$ normalized to $1$, i.e.,  $\sum_{j \in [h]} r_i = 1$, and sensitivity distribution $\gamma = (\gamma_i)_{i \in [h]}$, with $\gamma_1 < \gamma_2 < \dots < \gamma_h$. 
Let $\beta \ge 0$ be fixed and define $\epsilon = (\beta \gamma_i)_{i \in [h]}$. Then the $\epsilon$-stability ratio and the $\beta$-\deviation ratio are bounded by: 
\begin{align}
\eSR(\mathcal{I})  \le 1 + \beta \sum_{j = 1}^h r_j\gamma_j \quad\text{and}\quad
\bDR(\mathcal{I})  \le 1 + \beta \cdot \max_{j \in [h]} \bigg\{ \gamma_j \bigg( \sum_{p = j}^h r_p \bigg) \bigg\}.
\label{eq:comparison_epsilon}
\end{align}
Further, both bounds are tight for all distributions $r$ and $\gamma$.
\end{rtheorem}

%
%

\subsection{Proof of statement for $\epsilon$-$SR(\mathcal{I})$.}
\begin{proof}[Theorem \ref{thm:hetero}, $\epsilon$-stability ratio]
The statement for $\epsilon$-approximate equilibria can be proven almost similar as in Kleer and Sch\"afer \cite{Kleer2016} (the bound there is used for path deviations, but the proof extends directly to approximate equilibria). For completeness, we give the argument here.

For $j \in [k]$, let $\bar{P}_j$ be a path maximizing $l_P(x)$ over all flow-carrying paths $P \in \mathcal{P}$ of type $j$. Moreover, there exists a path $\pi$ such that $x_a \leq z_a$ and $z_a > 0$ for all $a \in \pi$ (see, e.g., Milchtaich \cite{Milchtaich2006}). We then have (this is also reminiscent of an argument by Lianeas et al. \cite{Lianeas2016}):
$$
l_{\bar{P}_j}(x) \leq (1+\beta \gamma_j) l_{\pi}(x) = (1 + \beta \gamma_j)\sum_{a \in \pi} l_a(x_a).
$$
Note that, by definition of the alternating path $\pi$, we have $x_a \leq z_a$ for all $a \in \pi$. Continuing with the estimate, we find $l_{\bar{P}_j}(x) \leq (1 + \beta \gamma_j)\sum_{a \in \pi} l_a(z_a)$ and thus
$$
C(x) \leq \sum_{j \in [h]} r_j l_{\bar{P}_j}(x) \leq  \sum_{j \in [h]} r_j (1 + \beta \gamma_j)\sum_{a \in \pi} l_a(z_a) = C(z) \bigg(\sum_{j \in [h]} r_j (1 + \beta \gamma_j)\bigg) 
$$
Since $\sum_{j \in [h]} r_j = 1$, we get the desired result. Note that we use $C(z) = \sum_{a \in \pi} l_a(z_a)$, which is true because there exists a flow-decomposition of $z$ in which $\pi$ is flow-carrying (here we use $z_a > 0$ for all $a \in \pi$).

 Tightness follows by considering an instance with arc set $\{0,1,\dots,h\}$ where the zeroth arc has latency $l_0(y) = 1$ and the arcs $j \in \{1,\dots,h\}$ have latency $l_j(y) = 1 + \beta \gamma_j$. An original Nash flow is given by $f^0 = (z_0,z_1,\dots,z_h) = (1,0,\dots,0)$, and an $\epsilon$-approximate Nash flow is given by $f^\epsilon = (x_0,x_1,\dots,x_h) = (0, r_1, r_2, \dots, r_h)$.\qed
\end{proof}

\subsection{Missing arguments for proof of $\beta$-deviation ratio.}

\begin{rtheorem}{Lemma}{\ref{lem:max_seq}}
Let $0 \leq \tau_{k-1} \leq \dots \leq \tau_1 \leq \tau_0$ and $c_i \geq 0$ for $i = 1,\dots,k$ be given. Then
$$
c_1 \tau_0 + \sum_{i = 1}^{k-1} (c_{i+1} - c_i)\tau_i \leq \tau_0 \cdot \max_{i=1,\dots,k} \{c_i\}.
$$
\end{rtheorem}
\begin{proof}
The statement is clearly true for $k = 1$. Now suppose the statement is true for some $k \in \N$. We will prove the statement for $k + 1$. \\
\textbf{Case 1: $c_{k+1} - c_{k} \leq 0$.} Then we have 
\begin{eqnarray}
c_1 \tau_0 + \sum_{i = 1}^{k} (c_{i+1} - c_i)\tau_i &\leq & c_1 \tau_0 + \sum_{i = 1}^{k-1} (c_{i+1} - c_i)\tau_i \text{ \ \ \ \  (using $\tau_k \geq 0$) }  \nonumber \\
&\leq & \tau_0 \cdot \max_{i=1,\dots,k} \{c_i\} \text{ \ \ \ \  (using induction hypothesis) } \nonumber \\
&\leq & \tau_0 \cdot \max_{i=1,\dots,k+1} \{c_i\}\text{ \ \ \ \  (using non-negativity of $\tau_0$) } \nonumber 
\end{eqnarray}\\
\textbf{Case 2: $c_{k+1} - c_{k} > 0$.}
\begin{eqnarray}
c_1 \tau_0 + \sum_{i = 1}^{k} (c_{i+1} - c_i)\tau_i &= & c_1 \tau_0 + \bigg[\sum_{i = 1}^{k-1} (c_{i+1} - c_i)\tau_i \bigg] + (c_{k+1}-c_k)\tau_{k}  \nonumber \\
&\leq & c_1 \tau_0 + \bigg[\sum_{i = 1}^{k-1} (c_{i+1} - c_i)\tau_i \bigg] + (c_{k+1}-c_k)\tau_{k-1} \text{ \ \  ($\tau_k \leq \tau_{k-1}$) }  \nonumber \\
&=& c_1 \tau_0 + \bigg[\sum_{i = 1}^{k-2} (c_{i+1} - c_i)\tau_i \bigg] + (c_{k+1}-c_{k-1})\tau_{k-1}  \nonumber \\
&\leq & \tau_0 \cdot \max_{i=1,\dots,k-2,k-1,k+1} \{c_i\} \text{ \ \ \ \  (using induction hypothesis) } \nonumber \\
&\leq & \tau_0 \cdot \max_{i=1,\dots,k+1} \{c_i\}\text{ \ \ \ \  (using non-negativity of $\tau_0$) } \nonumber 
\end{eqnarray}
Note that we apply the induction hypothesis with the set $\{c_1,\dots,c_{k-1},c_{k+1}\}$ of size $k$. \qed
\end{proof}

\subsubsection{Induction step to prove inequality (\ref{eq:induction}) in proof of Theorem \ref{thm:hetero}:}

The case $j = 1$ is precisely (\ref{eq:nash1}). Now suppose it holds for some $j$, then we have
\begin{eqnarray}
l_{P_{j+1}}(x) &\leq& l_{P_j}(x) + \gamma_{j+1}\delta_{P_j}(x) - \gamma_{j+1}\delta_{P_{j+1}}(x)  \ \ \ \ \text{(Nash condition for path $P_{j+1}$)} \nonumber \\
&\leq& l_{P_0}(x) + \gamma_1 \delta_{P_0}(x) + \bigg[\sum_{g = 1}^{j-1} (\gamma_{g+1} - \gamma_g)\delta_{P_g}(x)\bigg] - \gamma_j \delta_{P_j}(x) \nonumber \\
& & + \ \gamma_{j+1}\delta_{P_j}(x) - \gamma_{j+1}\delta_{P_{j+1}}(x) \nonumber \ \ \ \text{(induction hypothesis)} \nonumber \\
& = &l_{P_0}(x) + \gamma_1 \delta_{P_0}(x) + \bigg[\sum_{g = 1}^{j} (\gamma_{g+1} - \gamma_g)\delta_{P_g}(x)\bigg] - \gamma_{j+1} \delta_{P_{j+1}}(x), \nonumber
\end{eqnarray}
which shows the result for $j + 1$.\qed

\subsection{Bounded deviation model for $(0,\beta)$-path deviations and general sensitivity density distributions (single-commodity case).}
We use the  section  to explain the generalization of the bounds in Theorem \ref{thm:hetero} to more general sensitivity density distributions (see Corollary \ref{cor:het} further on). For sake of completeness, we adjust the definitions given in the description of the bounded deviation model (see Section \ref{sec:pre}).
We are given a single-commodity instance  $\mathcal{I} = (E,(l_e)_{e \in E},\mathcal{S})$ and $\beta \geq 0$ fixed. We consider the set
$$
\Delta(0,\beta) = \{ \delta =  (\delta_P)_{P \in \mathcal{S}} \ \big| \ 0 \leq \delta_P(f) \leq \beta l_P(f) \text{ for all } f \in \mathcal{F}(\mathcal{S}) \}
$$
of $(0,\beta)$-path deviation vectors $\delta= (\delta_P)_{P \in \mathcal{S}}$ where $\delta_P : \mathcal{F}(\mathcal{S}) \rightarrow \R_{\geq0}$ for all $P \in \mathcal{S}$. Moreover we have a Lesbesgue integrable (or just continuous) density function $\psi : \R_{\geq 0} \rightarrow \R_{\geq 0}$ over the the demand, which is normalized to $1$, i.e., $\int_0^\infty \phi(y)dy = 1$.  For feasible flow $f$, we have an indicator function 
$$
\mathbb{I} : \mathcal{S} \times [0,\infty) \rightarrow \{0,1\}
$$ 
which is $1$ if sensitivity $\gamma \in [0,\infty)$ is present on path $P$ in flow $f$, and $0$ otherwise. 

Define the \emph{deviated latency} of a path $P \in \mathcal{S}_i$ for sensitivity class $j \in [h_i]$ as $q_{P}^j(f) = l_P(f) + \gamma_{ij} \delta_P(f)$. We say that a flow $f$ is a \emph{$\beta$-\perturbed Nash flow} if there exist some $\beta$-bounded path deviations $\delta \in \Delta(\beta)$ such that
\begin{equation}
\forall \gamma [0,\infty), \forall s \in \mathcal{S} \text{ with }  \mathbb{I}(P,\gamma) = 1: \qquad q_{P}^j(f) \leq q_{P'}^j(f)  \ \ \forall P' \in \mathcal{S}.
\label{eq:nash}
\end{equation}
We define the \emph{$\beta$-\deviation ratio} $\bDR(\mathcal{I})$ as the maximum ratio $C(f^\beta)/C(f^0)$ of an $\beta$-\perturbed Nash flow $f^\beta$ and an original Nash flow $f^0$.


\medskip

\noindent The result of Theorem \ref{thm:hetero} for general sensitivity demand functions (as mentioned at the end of Section \ref{sec:het}) can now be stated as follows. 

\begin{corollary}\label{cor:het}

Let $\mathcal{I}$ be a network congestion game on a series-parallel graph with heterogeneous players, given by a Lebesgue integrable demand density function $\psi : \R_{\geq 0} \rightarrow \R_{\geq 0}$. Let $\beta \geq 0$ be fixed and let $\epsilon = (\beta\gamma_i)_{i \in [h]}$.  Then the $\epsilon$-stability ratio and the $\beta$-\deviation ratio are bounded by: 
\begin{align}
\eSR(\mathcal{I})  \le 1 + \beta \int_{0}^\infty y\cdot \psi(y)dy \quad\text{and}\quad
\bDR(\mathcal{I})  \le 1 + \beta  \sup_{t \in \R_{\geq 0}} \bigg\{ t \cdot \int_t^\infty \psi(y) dy \bigg\}.
\label{eq:comparison_epsilon_app}
\end{align}
Further, both bounds are asymptotically tight for all distributions $r$ and $\gamma$.

%
\end{corollary}
\begin{proof}
We first show that we can reduce to a discrete instance as considered in Theorem \ref{thm:hetero}. For every path $P \in \mathcal{S}$ we can set the sensitivity of all the flow on that path to 
$$
\gamma_P^* = \inf\{ \gamma : \mathbb{I}(P,\gamma) = 1 \}.
$$
In particular, this induces a demand distribution over a discrete (finite) valued sensitivity population (here we implicitly use the continuity of the latency functions). The upper bounds in Theorem \ref{thm:hetero} for the resulting sensitivity distributions will never be worse than the quantities in (\ref{eq:comparison_epsilon_app}). This concludes the proofs of the upper bounds (since the result now follows from the proof of Theorem \ref{thm:hetero}).

For tightness of the $\beta$-deviation ratio, fix $t \in [0,\infty)$, and let $T = \int_t^\infty \phi(y)dy$.  Consider the following instance on two arcs. We take $(l_1(y),\delta_1(y)) = (1,\beta)$ and take  $(l_2(y),\delta_2(y))$ with $\delta_2(y) = 0$ and $l_2(y)$ a strictly increasing function satisfying $l_2(0) = 1 + \epsilon'$ and $l_2(T) = 1 + t\cdot \beta$, where $\epsilon' < t \cdot \beta$. The (unique) original Nash flow is given by $z = (z_1,z_2) = (1,0)$ with $C(z) = 1$, and the (unique) deviated Nash flow $x$ is given by $x = (x_1,x_2) = (1-T,T)$ with $C(x) = 1 + \beta \cdot \gamma \cdot T$. Since this construction holds for all $t \in [0,\infty)$, the bound in (\ref{eq:comparison_epsilon_app}) is asymptotically tight. This holds since we can get arbitrarily close to the supremum if it is finite, and otherwise we can create a sequence of instances of which its ratio goes approaches infinity in case the supremum is not finite.

For tightness of the $\epsilon$-stability ratio, we create a discretized version of the construction in the proof of Theorem \ref{thm:hetero}. Fix some arbitrary $A> 0$. Choose $\alpha$ large enough so that 
$$
A = \int_\alpha^\infty \psi(y) dy
$$ 
For the density under $\alpha$, we discretize the interval $[0,\alpha]$ into the intervals $[k\cdot \epsilon', (k+1)\epsilon']$ where $k = 0,\dots,q$ for some $q = \alpha/\epsilon'$ (without loss of generality $q$ can be assumed to be integral, and it can be chosen as large as desired). We then consider the discrete demand distribution $r = (r_1,\dots,r_q,A)$, where 
$$
r_i = \int_{k\epsilon'}^{(k+1)\epsilon'} \psi(y) dy.
$$
We then create an instance with $q+1$ arcs, with latency functions $l_k(y) = (1 + k\epsilon')$ for $k = 1,\dots,q$ and $l_{q+1}(y) = 1 + A$. We now can use the same construction as in the proof of statement (\ref{eq:comparison_epsilon}), and by sending $A,\epsilon' \rightarrow 0$ we can get arbitrarily close to the value $1 + \beta \int_{0}^\infty y\cdot \psi(y)dy$. Note that the number of arcs in the instance goes to infinity as $A,\epsilon' \rightarrow 0$. \qed
\end{proof}

\section{Omitted material of Section \ref{sec:approx}}

\begin{rtheorem}{Proposition}{\ref{prop:equiv}}
Let $\mathcal{I}$ be a single-commodity congestion game with homogeneous players. $f$ is an $\epsilon$-approximate Nash flow for $\mathcal{I}$ if and only if $f$ is a $\beta$-\perturbed Nash flow for $\mathcal{I}$ with $\epsilon = \beta\gamma$. 
\end{rtheorem}
\begin{proof}
The first part of the proof is a special case of the proof of Proposition \ref{prop:inclusion}. Let $f$ be a $\beta$-\perturbed Nash flow for some set of path deviations $\delta \in \Delta(0,\beta)$. For $P \in \mathcal{S}$, we have
$$
l_P(f) \leq l_P + \gamma \delta_P(f) \leq l_{P'}(f) + \gamma \delta_{P'}(f)\leq (1+ \gamma \beta) l_{P'}(f) 
$$
where we use the non-negativity of $\delta_P(f)$ in the first inequality, the Nash condition for $f$ in the second inequality, and the fact that $\delta_{P'}(f) \leq \beta l_{P'}(f)$ in the third inequality. By definition, it now holds that $f$ is also a $\epsilon$-approximate equilibrium (for $\epsilon = \beta \gamma$).

Reversely, let $f$ be a $\epsilon$-approximate Nash flow (w.l.o.g., we can take $\gamma = 1$, so that $\beta = \epsilon$). We show that there exist $(0,\beta)$-path deviations $\delta_P$ such that $f$ is inducible with respect to these path deviations. Let $P_1,\dots,P_k$ be the set of flow-carrying paths under $f$, and assume without loss of generality that $l_{P_1}(f) \leq l_{P_2}(f) \leq \dots \leq l_{P_k}(f)$. We define $\delta_{P_i}(f) = l_{P_k}(f) - l_{P_i}(f)$ for $i = 1,\dots,k$. Using the Nash condition for the path $P_k$, we find
$$
\delta_{P_i}(f) = l_{P_k}(f) - l_{P_i}(f) \leq (1+ \beta)l_{P_i}(f) - l_{P_i}(f) = \beta l_{P_i}(f)
$$
which shows that these path deviations are feasible. Moreover, we take $\delta_{Q}(f) = \beta l_Q(f)$ for all the paths $Q \in \mathcal{P} \setminus \{P_1,\dots,P_k\}$ which are not flow-carrying under $f$. Now, let $i \in \{1,\dots,k\}$ be fixed. Then for any $j \in \{1,\dots,k\}$, we have
$$
l_{P_i}(f) + \delta_{P_i}(f) = l_{P_k}(f) = l_{P_j}(f) + \delta_{P_j}(f)
$$
and for any $Q \in \mathcal{P} \setminus \{P_1,\dots,P_k\}$, we have
$$
l_{P_i}(f) + \delta_{P_i}(f) = l_{P_k}(f) \leq (1+\beta)l_Q(f) = l_Q(f) + \delta_Q(f),
$$
using the Nash condition for the path $P_k$ and the definition of $\delta_Q(f)$. This shows that $f$ is indeed a $\beta$-deviated Nash flow.
\qed
\end{proof}


\subsubsection{Generalized Braess graph.} The \emph{$m$-th (generalized) Braess graph} $G^m = (V^m, A^m)$ is defined by $V^m = \{s,v_1,\dots,v_{m-1},w_1,\dots,w_{m-1},t\}$ and $A^m$ as the union of three sets: $E^m_1 = \{(s,v_j),(v_j,w_j),(w_j,t): 1 \leq j \leq m-1\}$, $E^m_2 = \{(v_j,w_{j-1}) : 2 \leq j \leq m\}$ and $E^m_3 = \{(v_1,t)\cup\{(s,w_{m-1}\}\}$ (see Figure \ref{fig:braess_5} for an example). \bigskip

\begin{rtheorem}{Theorem}{\ref{thm:braess_tight}}
Let $n = 2m$ be fixed and let $\mathcal{B}^m$ be the set of all instances on the generalized Braess graph with $n$ nodes. Then
$$
\sup_{\mathcal{I} \in \mathcal{B}^m} \eSR(\mathcal{I}) 
= 
\begin{cases}
\frac{1 + \epsilon}{1 - \epsilon \cdot (\lfloor n/2 \rfloor - 1)} & \text{if $\epsilon < \frac{1}{\lfloor n/2 \rfloor - 1}$,} \\
\infty & \text{otherwise.}
\end{cases}
$$
\end{rtheorem}

The proof of the upper bound for $\epsilon <  1/(\lfloor n/2 \rfloor - 1)$ follows from the discussion in the main text of Section \ref{sec:approx}. That is, we either have a path consisting solely of $Z$-edges (in which case we get a bound of $1 + \epsilon$ which follows from similar arguments as the proof of Theorem \ref{thm:hetero}) or we have an alternating path $q_i = 1$ for all sections of consecutive $X$-edges in the alternating path (in Figure \ref{fig:braess_5}, these are the edges with latency $T$). That is, every such section consists of one edge in $X$.

\begin{proof}[Lower bound for $\epsilon \geq  1/(\lfloor n/2 \rfloor - 1)$]
Let $\tau \geq 0$ be a fixed constant. Let $y_m : \R_{\geq 0} \rightarrow \R_{\geq 0}$ be a non-decreasing, continuous function with $y_m(1/m) = 0$ and $y_m(1/(m-1)) = \tau$.
We define 
$$
l^{m}_a(g) = \left\{ \begin{array}{ll}
(m - j)\cdot y_m(g) & \text{ for } a \in \{(s,v_j) : 1 \leq j \leq m-1\}\\
j \cdot y_m(g) & \text{ for } a \in \{(w_j,t) : 1 \leq j \leq m-1\}\\
1 & \text{ for } a \in \{(v_i,w_{i-1}): 2 \leq i \leq m-1\} \cup E_3\\
(1+\epsilon) + ((m-1)\epsilon - 1)\tau & \text{ for } a \in \{(v_i,w_i): 1 \leq i \leq m-1\} 
\end{array}\right.
$$
Note that $l_a^m$ is non-negative for all $a \in A$. This is clear for the first three cases, and the last case follows from the assumption $\epsilon \geq 1/(m-1)$.

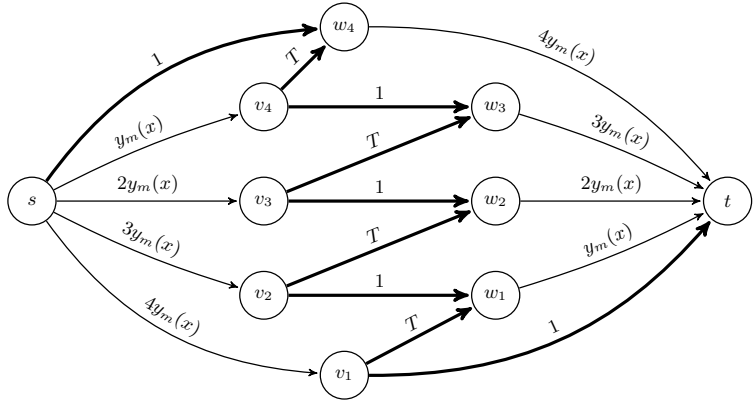
\begin{figure}[h!]
\centering
\scalebox{0.8}{
\begin{tikzpicture}[
  ->,
  >=stealth',
  shorten >=1pt,
  auto,
  semithick,
  every state/.style={circle,fill=white,radius=0.2cm,text=black},
]
r = 1
\begin{scope}
  \node[state]  (s)               					 {$s$};
  \node[state]  (v3) [right=3cm of s] 				 {$v_3$};
  \node[state]  (v4) [above=0.75cm of v3]  			 {$v_4$};
  \node[state]  (v2) [below=0.75cm of v3]			     {$v_2$};
  \node[state]  (v1) [below right=0.75cm and 0.75cm of v2] {$v_1$};
  \node[state]  (w4) [above right=0.75cm and 0.75cm of v4] {$w_4$};
  \node[state]  (w3) [right=3cm of v4] 				 {$w_3$};
  \node[state]  (w2) [right=3cm of v3] 				 {$w_2$};       
  \node[state]  (w1) [right=3cm of v2] 				 {$w_1$}; 
  \node[state]  (t)  [right=3cm of w2] 				 {$t$};
  
\path[every node/.style={sloped,anchor=south,auto=false}]
(s) edge[line width=1.5pt, bend left=25] 	node {$1$} (w4)            
(s) edge[bend left=5]  						node {$y_m(x)$} (v4)
(s) edge               						node {$2y_m(x)$} (v3)
(s) edge[bend right=5] 						node {$3y_m(x)$} (v2)
(s) edge[bend right=25]						node {$4y_m(x)$} (v1)
(v4) edge[line width=1.5pt]             	node {$T$} (w4)
(v4) edge[line width=1.5pt]         		node {$1$} (w3)
(v3) edge[line width=1.5pt]               node {$T$} (w3)
(v3) edge[line width=1.5pt]               node {$1$} (w2)
(v2) edge[line width=1.5pt]               node {$T$} (w2)
(v2) edge[line width=1.5pt]               node {$1$} (w1)
(v1) edge[line width=1.5pt]               node {$T$} (w1)
(w4) edge[bend left=25]node {$4y_m(x)$} (t)            
(w3) edge[bend left=5] node {$3y_m(x)$} (t)
(w2) edge              node {$2y_m(x)$} (t)
(w1) edge[bend right=5]node {$y_m(x)$} (t)
(v1) edge[line width=1.5pt, bend right=25]node {$1$} (t);
        
\end{scope}
\end{tikzpicture}}
\caption{The fifth Braess graph with $l_a^5$ on the arcs as defined in the proof of Theorem \ref{thm:braess_tight}, with $T = (1+\epsilon) + ((m-1)\epsilon - 1)\tau$. The bold arcs indicate the alternating path $\pi_1$.} 
\label{fig:braess_5}
\end{figure}

A Nash flow  $z = f^0$ is given by routing $1/m$ units of flow over the paths $(s,w_{m-1},t),(s,v_1,t)$ and the paths in $\{(s,v_j,w_{j-1},t) : 2 \leq j \leq m-1\}$. Note that all these paths have latency one, and the path $(s,v_j,w_j,t)$, for some $1 \leq m \leq j$, has latency $(1+\epsilon) + ((m-1)\epsilon - 1)\tau \geq 1$ for all $\tau \geq 0$ and $ \epsilon \geq 1/(m-1)$. We conclude that $C(z) = 1$.

An $\epsilon$-approximate Nash flow $x = f^\epsilon$ is given by routing $1/(m-1)$ units of flow over the paths in $\{(s,v_j,w_j,t) : 1 \leq j \leq m - 1\}$. Each such path $P$ then has a latency of 
$$
l_P(x) = (1+\epsilon) + ((m-1)\epsilon - 1)\tau + (m-j+j)\tau = (1+\epsilon)(1 + (m-1)\tau).
$$
Furthermore, we have $l_P(x) = (1+\epsilon)l_{P'}(x)$ for all $P' = (s,v_j,w_{j-1},t)$, where $2 \leq j \leq m-1$. The same argument holds for the paths $(s,w_{m-1},t)$ and $(s,v_1,t)$, so $x$ is indeed an $\epsilon$-approximate Nash flow. We have $C(x) = (1+\epsilon)(1 + (m-1)\tau)$. Sending $\tau \rightarrow \infty$ then gives the desired result. \qed
\end{proof}

\begin{proof}[Lower bound for $\epsilon < 1/(\lfloor n/2 \rfloor - 1)$]
Let $y_m : \R_{\geq 0} \rightarrow \R_{\geq 0}$ be a non-decreasing, continuous function with $y_m(1/m) = 0$ and $y_m(1/(m-1)) = \epsilon/(1 - \epsilon\cdot(m-1))$.
We define 
$$
l^{m}_a(g) = \left\{ \begin{array}{ll}
(m - j)\cdot y_m(g) & \text{ for } a \in \{(s,v_j) : 1 \leq j \leq m-1\}\\
j \cdot y_m(g) & \text{ for } a \in \{(w_j,t) : 1 \leq j \leq m-1\}\\
1 & \text{ for } a \in \{(v_i,w_{i-1}): 2 \leq i \leq m-1\} \cup E_3\\
1 & \text{ for } a \in \{(v_i,w_i): 1 \leq i \leq m-1\} 
\end{array}\right.
$$
Note that $l_a^m$ is non-negative for all $a \in A$.  

A Nash flow  $z = f^0$ is given by routing $1/m$ units of flow over the paths $(s,w_{m-1},t),(s,v_1,t)$ and the paths in $\{(s,v_j,w_{j-1},t) : 2 \leq j \leq m-1\}$. Note that all these paths have latency one, and the path $(s,v_j,w_j,t)$, for $1 \leq m \leq j$, has latency one as well. We conclude that $C(z) = 1$.

An $\epsilon$-approximate Nash flow $x = f^\epsilon$ is given by routing $1/(m-1)$ units of flow over the paths in $\{(s,v_j,w_j,t) : 1 \leq j \leq m - 1\}$. Each such path $P$ then has a latency of 
$$
l_P(x) = 1 + (m - j + j) \cdot \frac{\epsilon}{1 - \epsilon\cdot (m-1)} = \frac{1+\epsilon}{1 - \epsilon\cdot (m-1)}.
$$
Furthermore, for all $P' = (s,v_j,w_{j-1},t)$, where $2 \leq j \leq m-1$, we have
$$
(1+\epsilon)l_{P'}(x) = (1+\epsilon)\left(1 +  (m - 1) \cdot \frac{\epsilon}{1 - \epsilon\cdot (m-1)}\right) = l_P(x)
$$
The same argument holds for the paths $(s,w_{m-1},t)$ and $(s,v_1,t)$, so $x$ is indeed an $\epsilon$-approximate Nash flow. The result now follows since $C(x)/C(z) = (1+\epsilon)/(1 - \epsilon\cdot (m-1))$ as desired. \qed
\end{proof}


\subsection{Unboundedness of the $\epsilon$-stability ratio for matroid congestion games.}
We first show that the ratio between a $\beta$-approximate Nash flow and an original Nash flow can, in general, be unbounded.

\begin{theorem}\label{thm:matroid_unbounded}
Let $\epsilon \geq 0$ fixed and let $\mathcal{J} = (E,(l_e)_{e \in E},\mathcal{S})$ be a single-commodity matroid congestion game in which $\mathcal{S}$ contains all subsets of $E$ of precisely size $k$. Then
$$
\sup_{\mathcal{J}} \eSR(\mathcal{J})  \ \geq \ \left\{ \begin{array}{ll}
\frac{1 + \epsilon}{1 - \epsilon (k-1)} & \text{ \ \ \  if } \epsilon < 1/(k - 1) \\
\infty & \text{ \ \ \  if } \epsilon \geq  1/(k - 1),
\end{array}\right.
$$
For $k = 2$, the resulting lower bound $(1+\epsilon)/(1 - \epsilon)$ for $\epsilon < 1$ is tight (i.e., it is also an upper bound). This settles the case of $k=2$ completely.
\end{theorem}
\begin{proof}[Unboundedness]
Fix some $k \in \N$ and consider an instance with demand $r = 1$ and $E = \{e_0,\dots,e_k\}$, and, as stated above, let $\mathcal{S}$ be the family of all subsets of size $k$. Let $l_0(y) = 1$ be the constant latency function of resource $e_0$ and for $e \in \{e_1,\dots,e_k\}$, let $l_e(y)$ be a non-negative, continuous, increasing function satisfying 
$$
l_e((k-1)/k) = 1 \ \ \text{ and } \ \ l_e(1) = M,
$$ 
where $M$ is a constant to be determined later. 

A classical Nash flow $z = f^0$ is given by assigning $1/k$ units of flow to the strategies $E \setminus \{e_j\}$ for $j = 1,\dots,k$, so that the load on $e_0$ is $z_0 = 1$, and the load on resources $e_j$ is $x_{j} = (k-1)/k$. Every strategy has latency $k$ (since $l_e((k-1)/k) = 1$ for all resources $j \in \{1,\dots,k\}$), and therefore also $C(z) = k$.

We construct an $\epsilon$-approximate Nash flow $x = f^\epsilon$ as follows.The flow $x$ is defined by assigning all flow to the strategy $s = \{e_1,\dots,e_k\}$. For suitable choices of $M$, we show that $x$ is then indeed an $\epsilon$-approximate Nash flow. The set of other strategies is given by  $s^j = E \setminus \{e_j\}$ for $j \in \{1,\dots,k\}$. The Nash conditions for $x$ are then equivalent to
\begin{equation}\label{eq:nash_matroid}
l_s(x) = k\cdot M \leq (1 + \epsilon)[1 + (k-1)M] = l_{s_j}(x)
\end{equation}
for all $j \in \{1,\dots,k\}$. If $\epsilon \geq 1/(k-1)$, then $1 \leq \epsilon (k-1)$, and thus
$$
k \cdot M \leq 1 + \epsilon + (k-1)M + M \leq 1 + \epsilon + (k-1)M + \epsilon (k-1)M = (1+\epsilon)(1 + (k-1)M)
$$
for any non-negative $M$, i.e., the Nash conditions are always satisfied. Note that $C(x) = kM$ so that $C(x)/C(z) = M \rightarrow \infty$ as $M \rightarrow \infty$. Therefore the $\epsilon$-ratio is unbounded for $\epsilon \geq 1/(k-1)$.

If $\epsilon < 1/(k-1)$, then the Nash condition in (\ref{eq:nash_matroid}) is equivalent to 
$$
M \leq \frac{1 + \epsilon}{1 - \epsilon(k-1)}
$$
which is strictly positive since $\epsilon < 1/(k-1)$. In particular, by choosing $M = (1 + \epsilon)/(1 - \epsilon(k-1))$, we get $C(x)/C(z) = (k\cdot M)/k = (1 + \epsilon)/(1 - \epsilon(k-1))$ for $\epsilon < 1/(k-1)$. \qed
\end{proof}

For the case $k = 2$, the construction in the previous proof yields a lower bound of $(1+\epsilon)/(1 - \epsilon)$ if $\epsilon< 1$, and $\infty$ otherwise. For $\epsilon < 1$, we show that the bound $(1+\epsilon)/(1-\epsilon)$ is also an upper bound.

\begin{proof}[Upper bound for $k = 2$] 
Let $z = f^0$ be a classical Nash flow and let $x = f^\epsilon$ be a worst-case $\epsilon$-approximate Nash flow (both with normalized demand $r = 1$). 
We partition $E = X \cup Z$, where $Z = \{a \in E : z_a \geq x_a \text{ and } z_a > 0\}$ and $X = \{a \in E : z_a < x_a \text{ or } z_a = x_a = 0\}$.  

Let us first elaborate on the \emph{structure} of the flow $z$. One of the following cases holds.
\begin{enumerate}[i)]
\item There are two resources $e$ and $e'$ which are used by every player, that is, all flow is assigned to the strategy $\{e,e'\}$. By definition, both these resources will then be part of $Z$ (since $x_a \leq 1 = z_a$ for $a = e,e'$).
\item There is precisely one resource $e_0$ used by every player. Again, this resource will then be part of $Z$ by definition. Moreover, we then also have 
$$
l_{e_0}(1) = l_{e_0}(z_e) \leq l_{e'}(z_{e'})
$$ 
for every other resource $e'$. Also, all other resources $e'$ with $z_{e'} > 0$ have equal latency (which is true because of base exchange arguments), which in turn is greater or equal than that of $e_0$.
\item There are no resources which are used by every player. Then all flow-carrying resources in $z$ have equal latency (again because of base exchange arguments).
\end{enumerate}

\medskip

\noindent \textbf{Claim 1:} Without loss of generality, there is at most one resource in $Z$.
\begin{proof}
If there are at least two resources in $Z$, then there exists a pair of resources in $Z$, forming a strategy with a strictly positive amount of flow assigned to it in $z$ (by what was said above regarding the structure of $z$).\footnote{We emphasize that this is not necessarily true for all combinations of resources in $Z$.}
Using the fact that matroids are immune against the Braess paradox, as shown by Fujishige et al. \cite{Fujishige2015}, we can then carry out similar arguments as in the proof of Theorem \ref{thm:hetero} for the $\beta$-deviation ratio, but then for the more simple homogeneous case.\footnote{Phrased differently, the proof of Theorem \ref{thm:hetero} for the $\beta$-deviation ratio essentially relies on the fact that for the flow $z$ we can without loss of generality assume that all latency functions of resources in $Z$ are constant (by using immunity against the Braess paradox), and that there exists a strategy, with positive amount of flow assigned to it in $z$, only using resources in $Z$.} This would result in a bound of $1 + \epsilon \leq (1+\epsilon)/(1 - \epsilon)$. \qed
\end{proof}

 Let $P_{\max}$ be a strategy maximizing $l_P(x)$ over all flow-carrying strategies $P$ in $x$. \medskip

\noindent \textbf{Claim 2:} Without loss of generality, $P_{\max}$ uses only resources in $X$ (and not the resource $e_0$ in $Z$). 

\begin{proof} Suppose that every strategy with maximum latency is of the form $\{f,e_0\}$ for some resource $f \in E \setminus \{e_0\}$. Let $\{f_1,\dots,f_q\}$ be the set of all such resources $f$. It follows that, for any fixed $j = 1,\dots,q$, the strategy $\{e,f_j\}$, for some $e \in E \setminus \{e_0\}$, cannot be flow-carrying, otherwise,
\begin{eqnarray}
l_{e_0}(x_{e_0}) &\leq& l_{e_0}(z_{e_0})  \ \ \ \ \text{(definition of $Z$)} \nonumber \\
&\leq&  l_{e}(z_{e})   \ \ \ \ \ \ \text{(structure of Nash flow $z$ discussed above)} \nonumber \\
&\leq&    l_e(x_e)  \ \ \ \ \ \ \text{(definition of $X$)} \nonumber
\end{eqnarray}
which implies that $l_{e_0}(x_{e_0}) + l_{f_j}(x_{f_j}) \leq  l_e(x_e) + l_{f_j}(x_{f_j})$. Since $\{e_0,f_j\}$ is a strategy of maximum latency, then also $\{e,f_j\}$ is a strategy of maximum latency, which contradicts our assumption.

Moreover, not all players in $x$ can use resource $e_0$, otherwise $1 = x_{e_0} \leq z_{e_0}$, meaning that also all players use resource $e_0$ in $z$. But since $x_e > z_e$ for all other resources, which are in $X$ since we have only one resource in $Z$, this leads to a contradiction, since the total flow on all edges in $x$ is then higher than the total flow on all edges in $z$. 

Now, let $\{a,b\}$ be a flow-carrying strategy for some $a,b \in E\setminus \{e_0\}$, which exists by what was said above. For any flow-carrying strategy of the form $\{a,e\}$ for some $e \in E$, we have
$$
l_a(x_a) + l_e(x_e) \leq l_{P_{\max}}(x) \leq (1+\beta)l_P(x)
$$
for any other strategy $P$. We can then raise the value of $l_a(x_a)$ until the latency of some strategy of the form $\{a,e\}$ becomes tight with respect to $l_{P_{\max}}(x)$. Note that, since $a \in X$, we have $x_a > z_a$ and therefore this does not contradict the fact that $l_a$ is non-decreasing. More precisely, we replace $l_a$ by some function $\hat{l}_a(y)$ which is non-negative, continuous and non-decreasing, and that satisfies $\hat{l}_a(z_a) = l_a(z_a)$ and $\hat{l}_a(x_a) = l_a(x_a) + \alpha$, where $\alpha$ is the smallest value such that $l_a(x_a) + \alpha + l_e(x_e) = l_{P_{\max}}(x)$ for some flow-carrying resource of the form $\{e,a\}$. If $a \neq e_0$, we have found the desired result, and otherwise, we can use a similar argument as in the beginning of the claim to show that $\{a,b\}$ is also a flow-carrying strategy with maximum latency. Note that the social cost $C(x)$ of $x$ can only get worse if we replace the function $l_a$ by $\hat{l}_a$.\qed
\end{proof}  

We now use Claims 1 and 2 to establish the bound $(1+\epsilon)/(1-\epsilon)$ for $\epsilon < 1$. Let $P_{\max} = \{a,b\}$ for $a,b \in E \setminus \{e_0\}$ (which is justified because of Claim 2). The Nash conditions then give (by comparing $\{a,b\}$ with $\{a,e_0\}$ and $\{b,e_0\}$)
$$
l_a(x_a) + l_b(x_b) \leq (1+\epsilon)[l_a(x_a) + l_{e_0}(x_{e_0})]
$$
$$
l_a(x_a) + l_b(x_b) \leq (1+\epsilon)[l_b(x_b) + l_{e_0}(x_{e_0})]
$$
Adding up these inequalities, and rewriting, gives
\begin{equation}\label{eq:nash_matroid2}
(1 -\epsilon)[l_a(x_a) + l_b(x_b)] \leq (1+\epsilon)[l_{e_0}(x_{e_0})+ l_{e_0}(x_{e_0})]
\end{equation}
By definition of $Z$, we have $l_{e_0}(x_{e_0}) \leq l_{e_0}(z_{e_0})$, and, moreover, by the structure of the flow $z$ (as discussed in the beginning of the proof), we have $l_{e_0}(z_{e_0}) \leq l_{e}(z_{e})$ for any other resource with $z_e > 0$. On top of that, it also holds that $l_{e_0}(z_{e_0}) + l_{e}(z_{e}) = C(z)$, for any $e$ with $z_e > 0$, because of the fact that the demand is normalized to $r = 1$. Combining all this implies that $l_{e_0}(x_{e_0})+ l_{e_0}(x_{e_0}) \leq C(z)$. Using this observation, combined with (\ref{eq:nash_matroid2}), and the fact that $\epsilon < 1$, it follows that
$$
C(x) \leq l_a(x_a) + l_b(x_b) \leq \frac{1 + \epsilon}{1 - \epsilon} \cdot [l_{e_0}(x_{e_0})+ l_{e_0}(x_{e_0})] \leq \frac{1 + \epsilon}{1 - \epsilon} \cdot  C(z)
$$
which is the desired result. \qed
\end{proof}

\subsection{Proof of Theorem \ref{thm:matroid}}
\begin{rtheorem}{Theorem}{\ref{thm:matroid}}
Let $\mathcal{J} = (E,(l_e)_{e \in E},(\mathcal{S}_i)_{i \in [k]},(r_i)_{i \in [k]})$ be a matroid congestion game with homogeneous players. Let $\beta \ge 0$ be fixed and consider $\beta$-bounded basis deviations as defined above. 
Then the $\beta$-\deviation ratio is upper bounded by $\bDR(\mathcal{J}) \le 1 + \beta$. Further, this bound is tight already for $1$-uniform matroid congestion games. 

\end{rtheorem} 
\begin{proof}
The proof consists of establishing the following two claims.
\begin{enumerate}[i)]
\item We have  $l_e(x_e) \leq (1+\beta)l_e(z_e)$ for all $e \in E$ with $x_e > z_e$.
\item It holds that
$$
\sum_{e : x_e > z_e} (x_e - z_e)l_e(x_e) \leq (1 + \beta) \sum_{e : z_e \geq x_e} (z_e - x_e)l_e(x_e).
$$
\end{enumerate}
From these two claims, the result can be derived as follows. Rewriting the inequality in $ii)$, we find
$$
\sum_{e : x_e > z_e} x_el_e(x_e) \leq \sum_{e : x_e > z_e} z_el_e(x_e) + (1 + \beta) \sum_{e : z_e \geq x_e} (z_e - x_e)l_e(x_e).
$$
Adding $\sum_{e : z_e \geq x_e} x_el_e(x_e)$ to both sides of the inequality, and using the definition of $C(x)$, implies
$$
C(x)  \leq \sum_{e : x_e > z_e} z_el_e(x_e) + (1 + \beta)\sum_{e : z_e \geq x_e} z_e l_e(x_e) - \beta \sum_{e : z_e \geq x_e} x_el_e(x_e)
$$
In the first term, we now use that the fact that $l_e(x_e) \leq (1+\beta) l_e(z_e)$ for every $e$ in this summation, because of $i)$; in the second term, we use the fact that $l_e(x_e) \leq l_e(z_e)$ for all resources $e$ in the summation (which follows from $z_e \geq x_e$ and the fact that the latency functions are non-decreasing); the third term is left out using the fact that $\beta \sum_{e : z_e \geq x_e} x_el_e(x_e) \geq 0$ because of the non-negativity of the latency functions and $\beta$. This then implies that
$$
C(x) \leq  (1+\beta) \sum_{e : x_e > z_e} z_el_e(z_e) + (1 + \beta)\sum_{e : z_e \geq x_e} z_e l_e(z_e) - 0 = (1+\beta)C(z)
$$
which gives the desired result.

It now remains to show that $i)$ and $ii)$ holds. We first prove $i)$, of which the proof is of a similar nature as a proof of Fujishige et al. \cite{Fujishige2015}. In particular, using similar arguments as Lemma 3.2 \cite{Fujishige2015}, it can be shown that for every $e \in E$ with $x_e > z_e$, there exists an $f \in E \setminus \{e\}$ with $z_f > x_f$ such that 
$$
l_e(x_e) + \delta_e(x_e) \leq l_f(x_f) + \delta(x_f) \text{ \ \ \ and \ \ \ } l_f(z_f) \leq l_e(z_e).\footnote{We refer the reader to the proof of Lemma 3.2 \cite{Fujishige2015} for details. We omit it here for notational reasons.}
$$
The proof of this argument uses the fact that $x$ is a Nash flow w.r.t. the costs $l_e + \delta_e$, and $z$ a Nash flows w.r.t. to the latencies $l_e$.
It follows that
$$
l_e(x_e) \leq l_e(x_e) + \delta_e(x_e) \leq l_f(x_f) + \delta(x_f)  \leq (1+\beta)l_f(x_f) \leq (1 + \beta)l_f(z_f) \leq (1+\beta)l_e(z_e)
$$
using the fact that $0 \leq \delta_e(x_e) \leq \beta l_e(x_e)$ and the properties of the resources $e$ and $f$.

We now prove the second claim. We use a `variational inequality' argument similar to, e.g., the proof of Theorem 4 \cite{Lianeas2016}. Because of the fact that $x$ is $\beta$-deviated Nash flow, it follows that
$$
\sum_{e \in E} x_e(l_e(x_e)+ \delta_e(x_e)) \leq \sum_{e \in E} z_e(l_e(x_e) + \delta_e(x_e)).
$$
Rewriting gives
$$
\sum_{e: x_e > z_e} (x_e - z_e)(l_e(x_e)+ \delta_e(x_e)) \leq \sum_{e: z_e \geq x_e} (z_e - x_e)(l_e(x_e)+ \delta_e(x_e)).
$$
Using the fact that $0 \leq \delta_e(x_e) \leq \beta l_e(x_e)$ it follows that
$$
\sum_{e: x_e > z_e} (x_e - z_e)l_e(x_e) \leq \sum_{e: z_e \geq x_e} (z_e - x_e)(l_e(x_e)+ \delta_e(x_e)) \leq (1+\beta)\sum_{e: z_e \geq x_e} (z_e - x_e)l_e(x_e)
$$
which completes the proof.

Tightness follows, e.g., from a similar construction as in the last tightness construction of the proof of Theorem \ref{thm:hetero} (if applied to a homogeneous population). \qed
\end{proof}

\end{document}